\documentclass[11pt]{article}
\usepackage[letterpaper, margin=1in]{geometry}

\usepackage{graphicx}
\usepackage{amsthm,amsmath,amssymb}
\usepackage{todonotes}
\usepackage{geometry}
\usepackage{algorithm}
\usepackage[noend]{algpseudocode}
\usepackage{subcaption}
\usepackage{xspace}
\usepackage{hyperref}
\usepackage{thmtools}
\usepackage{thm-restate}
\usepackage{wrapfig}

\newtheorem{theorem}{Theorem}[section]

\newtheorem{lemma}[theorem]{Lemma}

\newtheorem{claim}{Claim}[theorem]
\newtheorem{observation}[theorem]{Observation}

\newcommand*{\claimproofs}{Proof of the Claim.\quad}
\newenvironment{claimproof}[1][\claimproofs]{\begin{proof}[#1]}{\end{proof}}

\newcommand{\NP}{\textsf{NP}}

\newcommand{\yes}{\texttt{yes}}
\newcommand{\no}{\texttt{no}}

\newcommand{\cP}{\mathcal{P}}
\newcommand{\cF}{\mathcal{F}}
\newcommand{\cS}{\mathcal{S}}

\newcommand{\RR}{\textsf{RR}}

\begin{document}

\title{\Large A Quasi-Polynomial Time Algorithm for \\ 3-Coloring Circle Graphs}
    \author{
    Ajaykrishnan E S\thanks{University of California, Santa Barbara, USA. 
    Email: es.ajaykrishnan@gmail.com, daniello@ucsb.edu. 
    Supported by NSF grant 2505099: Collaborative Research: AF Medium: Structure and Quasi-Polynomial Time Algorithms.}
    \and
    Robert Ganian\thanks{Technische Universität Wien, Austria. 
    Email: rganian@gmail.com. 
    Supported by Austrian Science Fund (FWF) Projects Y1329 and COE12, and the WWTF Vienna Science and Technology Fund Project ICT22029.}
    \and
    Daniel Lokshtanov\footnotemark[1]
    \and
    Vaishali Surianarayanan\thanks{University of California, Santa Cruz, USA. 
    Email: vaishalisurianarayanan@gmail.com. 
    Supported by the UCSC Chancellor’s Postdoctoral Fellowship.}
    }

\date{}

\maketitle

\begin{abstract}
A graph $G$ is a {\em circle} graph if it is an intersection graph of chords of a unit circle. 
We give an algorithm that takes as input an $n$ vertex circle graph $G$, runs in time at most $n^{O(\log n)}$ and finds a proper $3$-coloring of $G$, if one exists. 
As a consequence we obtain an algorithm with the same running time to determine whether a given ordered graph $(G, \prec)$ has a $3$-page book embedding. 
This gives a partial resolution to the well known open problem of 
Dujmovi\'{c} and Wood [Discret. Math. Theor. Comput. Sci. 2004],
Eppstein [2014],
and Bachmann, Rutter and Stumpf [J. Graph Algorithms Appl. 2024]
of whether 3-{\sc Coloring} on circle graphs admits a polynomial time algorithm.  
\end{abstract}


\section{Introduction}

A \emph{circle graph} is an intersection graph of chords in a unit circle, i.e., a graph whose vertices can be associated with chords in a circle and where two vertices are adjacent if and only if their corresponding chords intersect. Circle graphs are well-studied from the graph-theoretic as well as algorithmic~\cite{BachmannRS24,BrijderT22,BrucknerRS24,ChaplickFK19,Damian-IordacheP00,GaborSH89} perspectives, and are known to be recognizable in almost-linear time~\cite{GioanPTC14}. 
On a separate note, given a (general) graph $G$ and a linear order $\prec$ of $V(G)$, we say that $(G,\prec)$ has a \emph{$q$-page book embedding} if the edges of $G$ can be partitioned into $q$ sets $E_1,\dots,E_q$ (called \emph{pages}) such that no pair of edges $ab$, $cd$ mapped to the same page form the configuration $a\prec c \prec b \prec d$. Intuitively, one often visualizes pages as half-planes whereas a book embedding ensures that the edges can be drawn on their assigned pages in a crossing-free manner. Book embeddings have been~\cite{chung1987embedding,Yannakakis86,Yannakakis89} and remain~\cite{BachmannRS24,BhoreGMN20,DepianFGN24,GanianMOPR24} the focus of extensive research, both with and without a fixed choice of the ordering $\prec$.


\begin{figure}[htbp]
    \centering
    \hspace{4pt}
    \includegraphics[width=0.45\textwidth]{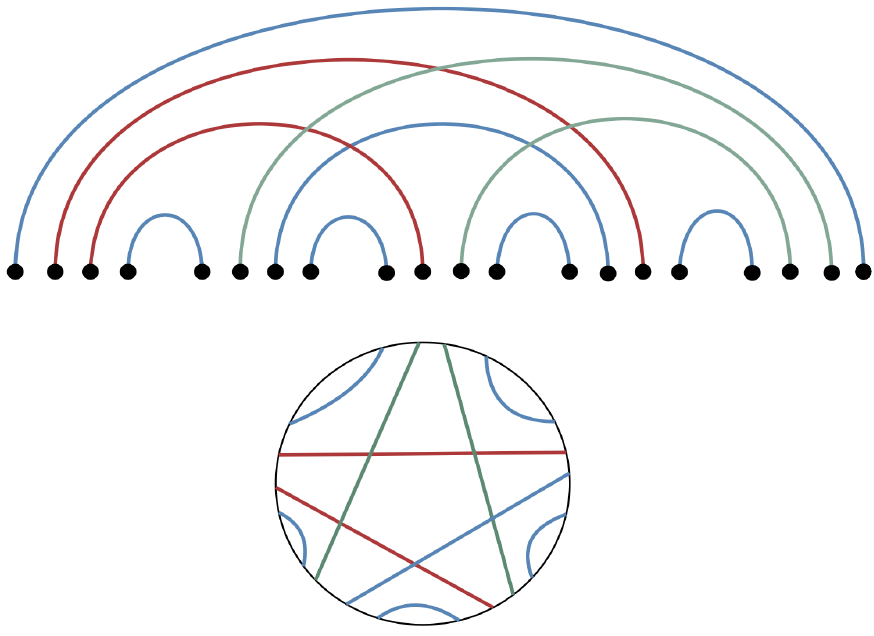}
    \caption{Correspondence between 3-page book embedding and circle graph 3-coloring (Source~\cite{BachmannRS24})}
    \label{fig:book_thickness}
\end{figure}

There is a direct and folklore correspondence between computing a $q$-page book embedding (for a given pair $(G,\prec)$, where $G$ is a general graph) and computing a proper $q$-coloring of a circle graph (for a given circle graph $H$)~\cite[Page 2]{BachmannRS24}, refer Figure~\ref{fig:book_thickness} (sourced from \cite{BachmannRS24}). Naturally, we can determine whether a circle graph admits a proper $2$-coloring in polynomial time, but the same question for $4$ colors is known to be \NP-hard~\cite{Unger88}.
In this article, we focus on the in-between case of \textsc{3-Coloring} on circle graphs, which has received considerable attention and has a curious history. 

In 1992, Unger~\cite{Unger92} claimed a polynomial-time algorithm for the problem, but the conference version of that result was missing many crucial details and referenced Unger's thesis for the full proof. Unfortunately, the thesis itself is written in German, is not available online, and the proof in question has not been reproduced since. This unsatisfactory state of affairs was pointed out by several senior researchers in the graph drawing community, including Eppstein~\cite{Eppsteinblog} who claimed the problem should be considered open; in their earlier survey, Dujmović and Wood~\cite{DujmovicPW04} also listed the problem as open. In their recent paper~\cite{BachmannRS24} on this specific topic, Bachmann, Rutter and Stumpf identified a concrete and seemingly critical flaw in Unger's original algorithm and also provided a more detailed recapitulation of the problem's history. As of now, the complexity status of \textsc{3-Coloring} on circle graphs remains a prominent open question.

In this article, we present a quasi-polynomial-time algorithm with runtime $n^{O(\log n)}$ for \textsc{3-Coloring} on circle graphs, which can be seen as a strong indication that the problem is unlikely to be \NP-hard. In fact, our algorithm also solves the more general \textsc{List 3-Coloring} problem. While the result is non-trivial, the proof does not involve complicated machinery and is fully reproducible. Crucially, it represents tangible progress on a problem which has remained ``in limbo'' for more than three decades.

\subparagraph{Related Work.}
The study of book embeddings dates back to the seventies of the previous century~\cite{Kainen74},
among others due to their applications in bioinformatics, VLSI, and parallel computing (see, e.g.,~\cite{chung1987embedding,DujmovicPW04,Haslinger1999}). 
A classical result of Yannakakis establishes that every planar graph $G$ admits a $4$-page book embedding (for some ordering $\prec$), and this result is tight~\cite{Yannakakis86,Yannakakis89}.
Similar results have since then been obtained for a variety of other graph classes~\cite{BekosLGGMR20,BekosLGGMR24,WoodD11}.
However, on the computational side deciding whether an input graph $G$ has an ordering $\prec$ admitting a book embedding with two pages is already \NP-complete, as this is equivalent to the \NP-complete problem of testing whether $G$ is a subgraph of a planar Hamiltonian graph~\cite{BernhartK79,chung1987embedding}. This problem was recently shown to admit a subexponential algorithm~\cite{GanianMOPR24}.

When $\prec$ is provided as part of the input, the aforementioned polynomial-time equivalence between computing book embeddings and coloring circle graphs applies. Unger's claimed polynomial-time algorithm~\cite{Unger92} for solving the open case of 3-coloring circle graphs operated by first constructing a 3-\textsc{Sat} formula that is satisfiable if and only if the circle graph admits a 3-coloring. Bachmann, Rutter and Stumpf recently showed that the proposed construction can, however, produce a satisfiable formula even though the input graph is not 3-colorable~\cite{BachmannRS24}. 
The second part of the claimed algorithm~\cite{Unger92} was then supposed to verify the satisfiability of the produced 3-\textsc{Sat} formulas via a backtracking strategy that relied on the specific structure of these formulas. In the same recent publication~\cite{BachmannRS24}, Bachmann, Rutter and Stumpf showed that the proposed backtracking strategy is flawed as well---both in terms of correctness and running time behavior.

\subparagraph{Proof Overview.}
We will assume that a chord diagram (a representation of $G$ as the intersection graph of chords on a unit circle, where all chords are straight lines with distinct endpoints) for $G$ is provided as part of the input; if not, we can compute it using the almost-linear time algorithm of Gioan, Paul, Tedder and Corneil~\cite{GioanPTC14}. 
We will solve the more general {\sc List $3$-Coloring} where, in addition to $G$ we are given as input a function $S : V(G) \rightarrow 2^{\{\texttt{red},\texttt{blue},\texttt{green}\}}$ and the task is to find a proper coloring $\texttt{col} : V(G) \rightarrow \{\texttt{red},\texttt{blue},\texttt{green}\}$ such that $\texttt{col}(v) \in S(v)$ for every vertex $v \in V(G)$.
Working with lists allows us to conveniently delete a vertex $v$ from the graph if we already have decided which color $v$ should get.  
We may then remove $v$'s color from the lists of its neighbors, and delete $v$ from $G$.

The key observation that we use throughout the algorithm is as follows.  
A $4$-{\em partition}, called {\em circle partition} in our main text - see Figure~\ref{fig:circle_definitions} - is a partition of the unit circle into $4$ arcs $A$, $B$, $C$, $D$ encountered in this order when moving clockwise around the circle. 
We can classify the chords of the circle, that is the vertices of $G$, according to which part of the partition the endpoints of the chord lie in. So an $A$-$C$ chord has one endpoint in $A$ and the other in $C$.
Every $A$-$C$ chord crosses every $B$-$D$ chord. Thus, if $G$ is $3$-colorable, either all of the $A$-$C$ chords or all of the $B$-$D$ chords must get the same color. 
{\em Branching} on the $4$-partition $(A, B, C, D)$ refers to recursively solving the $6$ instances resulting from assigning to all of the $A$-$C$ chords or to all of the $B$-$D$ chords one of the $3$ possible colors (and deleting these chords after appropriately updating the lists of neighbor vertices). 
If the guess results in a coloring that violates the lists of any of the colored vertices, or the list of an uncolored vertex becomes empty, the resulting instance is a \no\ instance and is skipped. 
It is clear that $G$ has a $3$-coloring respecting
the lists if and only if one of the resulting $6$ instances does as well. 

The entire algorithm relies on branching on (carefully chosen) $4$-partitions. The $4$-partitions are chosen in such a way that after moving at least $1 + \log_{4/3} n$ steps down any path in the recursion tree, the graph becomes disconnected and the largest connected component has at most $\frac{3n}{4}$ vertices. 
Since distinct connected components can be handled independently, this leads to the running time $T(n)$ satisfying a recurrence of the form $T(n) \leq n^{O(1)} \cdot T(\frac{3n}{4})$,
which solves to $T(n) \leq n^{O(\log n)}$.
Somewhat similar branching strategies have recently been applied for {\sc Independent Set} and also 3-{\sc Coloring} on various classes of graphs~\cite{GartlandL20,GartlandLPPR21,PilipczukPR21}.

We now sketch how to carefully choose $4$-partitions in order to achieve the goal above. 
We first pick a partition $(L, T, R, B)$ (a mnemonic for left, top, right, bottom) where each of the four parts contain precisely one quarter of chord endpoints, see Figure~\ref{fig:branching_step_1}. We branch on the $4$-partition $(L, T, R, B)$. In each of the resulting recursive calls, either there are no $L$-$R$ chords or there are no $T$-$B$ chords. Without loss of generality we may assume that there are no $L$-$R$ chords. 
Our next goal is to separate $L$ from $R$. The idea is to branch on $4$-partitions $(A, B, C, D)$ in such a way that the resulting vertex deletions allow us to ``grow'' $L$ and $R$, ``shrink'' $T$ and $B$ until eventually no chord has an endpoint in $T \cup B$, while maintaining the invariant that there are no $L$-$R$ chords. 
More formally, in addition to $G$, $S$ and its chord diagram, this subroutine will take as input a $4$-partition $(L', T', R', B')$ such that $L \subseteq L'$, $R \subseteq R'$, $T' \subseteq T$, $B' \subseteq B$ and there are no $L'$-$R'$ chords. 
Initially we call the subroutine with the $4$-partition $(L, T, R, B)$.

Given as input the $4$-partition $(L', T', R', B')$, 
%
%
the subroutine selects the arc out of $\{T', B'\}$ that contains the most endpoints. Without loss of generality it is $T'$.
We split $T'$ into two sub-arcs $T_L'$ and $T_R'$ such that $T_L'$ and $T_R'$ contain the same number of endpoints (plus or minus one), $T_L'$ is incident to arc $L$ and $T_R'$ is incident to arc $R$.
We branch on the $4$-partition $(L', T_L', T_R', R' \cup B')$, see Figure~\ref{fig:branching_step_2}.
After this branching step, either there are no $L'-T_R'$-chords, or there are no $R'-T_L'$-chords. 
If there are no $L'-T_R'$-chords then the $4$-partition
$L'$, $T_L$, $R' \cup T_R'$, $B'$ satisfies the invariant of the subroutine (that there are no $L'-(R' \cup T_R')$-chords), and the number of endpoints in $T_L \cup B'$ is at most $\frac{3}{4}$ of the number of endpoints in $T' \cup B'$.
Similarly, if there are no $R'-T_L'$-chords then the $4$-partition
$L' \cup T_L'$, $T_R$, $R$, $B'$ satisfies the invariant of the subroutine, and the number of endpoints in $T_L \cup B'$ is at most $\frac{3}{4}$ of the number of endpoints in $T' \cup B'$.
Thus, with every step down in the recursion tree of the subroutine the number of endpoints in $T' \cup B'$ drops by a constant factor.

Therefore, at recursion depth $\log_{4/3} n$, no chords have an endpoint in $T' \cup B'$.
At this point, since there are no $L'-R'$ chords, it follows that every connected component of (what remains of) $G$ either does not contain any chords with an endpoint in $L$, or does not contain any chords with an endpoint in $R$. Since $L$ contains $\frac{n}{2}$ endpoints, the set of vertices with at least one endpoint in $L$ has size at least $\frac{n}{4}$, and the same lower bound holds for $R$.
Thus every component of $G$ has size at most $\frac{3n}{4}$, as was our goal. Hence the running time of the algorithm is upper bounded by the recurrence $T(n) \leq 6^{\log_{4/3} n + O(1)} \cdot T(\frac{3n}{4}) \leq n^{O(1)} \cdot T(\frac{3n}{4})$, which in turn is upper bounded by $n^{O(\log n)}$, giving the claimed upper bound for the running time of the algorithm.

\section{Preliminaries}
We assume basic familiarity with graph theory, and in particular the notions of \emph{vertices} and \emph{edges} of a graph $G$ (denoted $V(G)$ and $E(G)$, respectively).
An \emph{independent set} of a graph $G$ is a subset of vertices of $G$ such that no two vertices that belong to it are adjacent to each other.

A \emph{chord diagram} $(H,\alpha)$ of a graph $G$ consists of a set $H$ of straight-line chords in the unit circle along with a bijection $\alpha$ between $V(G)$ and $H$, such that $ab\in E(G)$ if and only if the chords $\alpha(a)$ and $\alpha(b)$ intersect. Throughout the paper and without loss of generality, we assume that no pair of chords in a chord diagram share the same endpoint~\cite{EsperetO09,EsperetS20}.
Graphs which admit a chord diagram are called \emph{circle} graphs. It is known that the almost linear-time algorithm for recognizing circle graphs can also output a chord diagram as a witness~\cite{GioanPTC14}. 

A \emph{$3$-coloring} of a graph $G$ is a mapping from its vertices into \{\texttt{red}, \texttt{blue}, \texttt{green}\}, such that the set of vertices assigned each color forms an independent set.
We put our problem of interest on formal footing below.

\begin{center}
	\noindent\fbox{
		\begin{minipage}{0.75\textwidth}
			\textsc{Circle Graph List $3$-Coloring} \\
			{\bf{Input:}} An $n$-vertex circle graph $G$ along with its chord diagram $(H,\alpha)$ and a color list $S: V(G) \rightarrow 2^{\{\texttt{red}, \texttt{blue}, \texttt{green}\}}$.  \\
			{\bf{Question:}} Does $G$ admit a $3$-coloring $\texttt{col}$ such that for each $v\in V(G)$, $\texttt{col}(v)\in S(v)$?
		\end{minipage}}
\end{center}

We say an instance $I$ is a \yes\ instance if it admits a $3$-coloring $\texttt{col}$ such that for each $v\in V(G)$, $\texttt{col}(v)\in S(v)$ and call it a \no\ instance otherwise. We call such a 3-coloring a valid 3-coloring. We use $H(I)$ to denote the set of chords $H$ in $I$. We use $n_I$ to denote $|V(G)|=|H(I)|$, the number of vertices and chords in $I$. When $I$ is clear from context, we will drop the subscript. We say that an instance $I=(G,(H,\alpha),S)$ is a \emph{subinstance} of $I'=(G',(H',\alpha'),S')$ if the former can be obtained from the latter by deleting vertices and removing colors from the lists---formally, $G$ is an induced subgraph of $G'$, $H\subseteq H'$, $\alpha\subseteq \alpha'$, and for each $v\in V(G')$ it holds that $S(v)\subseteq S'(v)$.  For subinstances, note that the same chord diagram is maintained apart from removing the chords of deleted vertices.

\paragraph{Circle, Arcs, Circle Partition:} We call a connected region on the boundary of a unit circle an \emph{arc}, and a straight line segment between two points on the boundary of a unit circle a \emph{chord}. For two disjoint arcs $A$ and $B$ of a unit circle, we call a chord an {\em $A$-$B$ chord} if it has one endpoint in $A$ and one in $B$. Likewise, we call a chord an $A$-$A$ chord if it has both endpoints in $A$.

A {\em circle partition $\cP=(L,T,R,B)$} is a partition of the unit circle into four, possibly empty, arcs $L,T,R,B$ where $L\cup T$, $T\cup R$, $R\cup B$, and $B\cup L$ are also arcs. 
For convenience, to avoid writing $\cP=(L,T,R,B)$, given a circle partition $\cP$ we use $L(\cP), T(\cP),$ $R(\cP), B(\cP)$ to refer to the arcs $L,T,R,B$. See Figure~\ref{fig:circle_definitions} for an example of these terminologies. 

\begin{figure}[htbp]
  \centering
  
  \begin{subfigure}[b]{0.31\textwidth}
    \includegraphics[width=\textwidth]{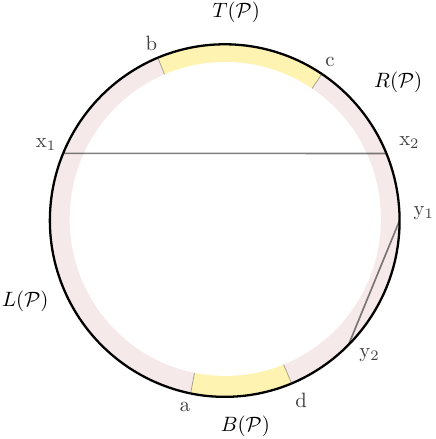}
    \caption{}
    \label{subfig:prelima}
  \end{subfigure}
  \hspace{20pt}
  \begin{subfigure}[b]{0.31\textwidth}
    \includegraphics[width=\textwidth]{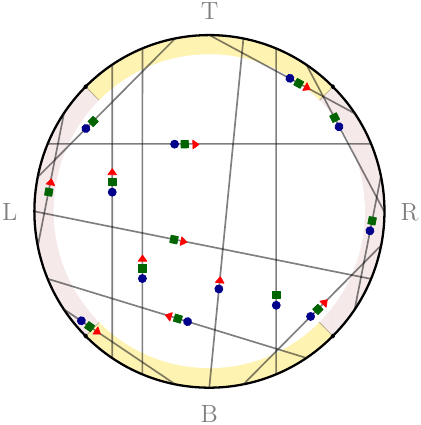}
    \caption{}
    \label{subfig:prelimb}
  \end{subfigure}
  
  \caption{Subfigure~\ref{subfig:prelima} shows a unit circle with arcs $ab$, $bc$, $cd$, and $da$ forming a circle partition $\mathcal{P}$, where $L(\mathcal{P}) = ab$, $T(\mathcal{P}) = bc$, $R(\mathcal{P}) = cd$, and $B(\mathcal{P}) = da$. Also, $x_1x_2$ and $y_1y_2$ are $L(\cP)$–$R(\cP)$ and $R(\cP)$–$R(\cP)$ chords, respectively. Subfigure~\ref{subfig:prelimb} depicts a chord representation of a \textsc{Circle Graph List 3-Coloring} instance $I$ with a circle partition $\cP = (L,T,R,B)$, where lists are indicated using {\sf red} triangles, {\sf green} squares, and {\sf blue} circles.}
  \label{fig:circle_definitions}
\end{figure}

\section{The Algorithm}

\begin{theorem}\label{thm:main_thm}
    \textsc{Circle Graph List 3-Coloring} 
    admits an algorithm running in time $n^{O(\log n)}$. 
\end{theorem}




Let $I$ be an instance of \textsc{Circle Graph List 3-Coloring}. On a high level, we will design an algorithm that in order to solve the instance $I$, recursively solves polynomially many subinstances of $I$ having size bounded by a constant fraction of the size of $I$. We start by defining a notion that will help us break an instance into subinstances based on its chord diagram. 


\paragraph{Fully-Separated Instances.}We call a tuple $(I,\cP)$, where $I$ is an instance of \textsc{Circle Graph List 3-Coloring} and $\cP$ is a circle partition, a {\em fully-separated instance} if it satisfies (i) $T(\cP)=B(\cP)=\emptyset$ and (ii) $H(I)$  contains no $L(\cP)$-$R(\cP)$ chords. See Figure~\ref{fig:fully_separated} for an example.

As the name suggests, these fully separated instances will help us break instances into disjoint subinstances. Given a fully-separated instance $(I,\cP)$, we denote by $I_{L(\cP)}$ the subinstance of $I$ obtained by removing all but $L(\cP)$-$L(\cP)$ chords and their corresponding vertices from $I$. Similarly, we denote by $I_{R(\cP)}$ the subinstance of $I$ obtained by removing all but $R(\cP)$-$R(\cP)$ chords and their corresponding vertices from $I$. See Figure~\ref{fig:fully_separated} for an example of these subinstances.

\begin{figure}[htbp]
    \centering
    
    \begin{subfigure}[b]{0.31\textwidth}
    \includegraphics[width=\textwidth]{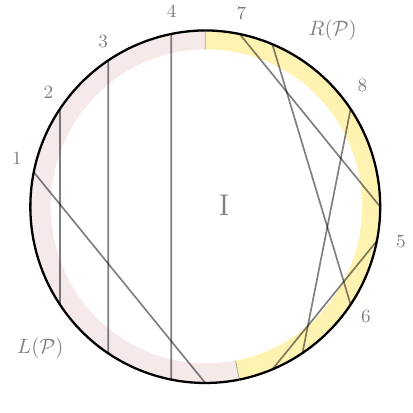}
    \caption{}
    \label{subfig:2a}
    \end{subfigure}
    \hspace{5pt}
    \begin{subfigure}[b]{0.307\textwidth}
    \includegraphics[width=\textwidth]{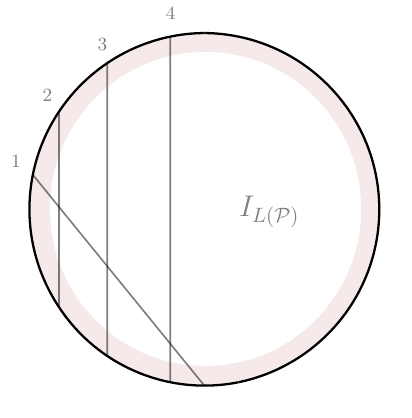}
    \caption{}
    \label{subfig:2b}
    \end{subfigure}
    \begin{subfigure}[b]{0.313\textwidth}
    \includegraphics[width=\textwidth]{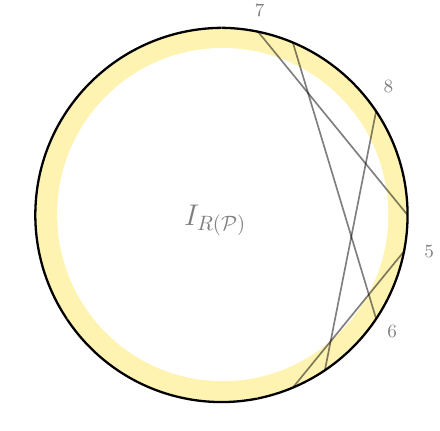}
    \caption{}
    \label{subfig:2c}
    \end{subfigure}
    
    \caption{Subfigure~\ref{subfig:2a} is a chord diagram corresponding to a fully-separated instance $(I,\cP)$ (arcs $T(\cP)$ and $B(\cP)$ are not shown since they are empty). Subfigures~\ref{subfig:2b} and~\ref{subfig:2c} show chord diagrams corresponding to the subinstances $I_{L(\cP)}$ and $I_{R(\cP)}$ of $I$ respectively.}
    \label{fig:fully_separated}
\end{figure}

Since $I$ has no $L(\cP)$-$R(\cP)$ chords and $\cP$ is a circle partition with $T(\cP)=B(\cP)=\emptyset$, we have the following observation relating $I$ with $I_{L(\cP)}$ and $I_{R(\cP)}$.
\begin{observation}\label{obs:fully-sep-split}
    For a fully-separated instance $(I,\cP)$, $I$ is a \yes\ instance if and only if both $I_{L(\cP)}$ and $I_{R(\cP)}$ are \yes\ instances.
\end{observation}

Let $I$ be an instance of \textsc{Circle Graph List 3-Coloring} that we want to solve. 
If we can find a family of fully-separated instances such that $I$ is a \yes\ instance if and only if at least one subinstance in the family is a \yes\ instance, Observation~\ref{obs:fully-sep-split}, will allow us to recursively solve the problem on the disjoint subinstances $I'_{L(\cP)}$ and $I'_{R(\cP)}$ corresponding to instances $I'$ belonging to the family. We will show that such a ``\emph{good}'' family can be computed in $n^{\mathcal{O}(1)}$ time. 
%
In the family we compute, each $(I',\cP)$ will satisfy (i) $I'$ is a subinstance of $I$ and (ii) $n_{I'_{L(\cP)}}, n_{I'_{R(\cP)}} \leq \frac{3}{4} n_I$. We then use this family to design an algorithm that solves $I$ by recursively solving $n^{O(1)}$ many subinstances of $I$ having size at most $\frac{3}{4}n_I$. This yields a running time bound of $n^{O(\log n)}$ as the number of recursive calls is $n^{O(1)}$ and the depth of the recursion is $O(\log n)$.  
%
We now state a lemma that shows how to find such a family, phrased in a slightly different way. 
\begin{lemma}
\label{lem:family-fully-seperated-instances}
    There exists an algorithm that takes as input an instance $I$ of \textsc{Circle Graph List 3-Coloring}, runs in time $n^{O(1)}$, and returns a family containing at most $ n^{O(1)}$ fully-separated instances such that:
    \begin{enumerate}
        \item  $I$ is a \yes\ instance if and only if the family contains a fully-separated instance $(I',\cP')$ where $I'$ is a \yes\ instance.
        \item For each fully-separated instance $(I',\cP')$ in the family:
        \begin{enumerate}
            \item $I'$ is a subinstance of $I$
            \item The arcs $L(\cP')$ and $R(\cP')$ contain at least $\lfloor n/2 \rfloor$ endpoints of chords in $H(I)$ 
        \end{enumerate}     
    \end{enumerate}
\end{lemma}

We defer the proof of Lemma~\ref{lem:family-fully-seperated-instances} to Subsection~\ref{sub:lemmaproof}. Before proceeding there, we use the lemma to establish our main result, Theorem~\ref{thm:main_thm}. 
\begin{proof}[Proof of Theorem~\ref{thm:main_thm}] We first present our algorithm that takes an input an instance $I$ of \textsc{Circle Graph List 3-Coloring} and returns \yes\ if $I$ is an \yes\ instance and \no\ otherwise.\\ \\
   \textsf{ALG-3-Color-Circle-Graph ($I$):}
   \begin{enumerate}
       \item  Let $n:=n_I$. If $n\leq 8$, brute force over each possible 3-coloring of $G$ respecting the list $S$ and return \yes\ if there exists a proper 3-coloring of $G$, else return \no. If $n>8$, proceed to the next step.
       \item Obtain a family $\cF$ of fully-separated instances using Lemma~\ref{lem:family-fully-seperated-instances} on $I$. 
       \item  For each fully-separated instance $
    (I',\cP')\in \cF$
       \begin{enumerate}
           \item Recursively call the algorithm, \textsf{ALG-3-Color-Circle-Graph}, on $I'_{L(\cP')}$ and $I'_{R(\cP')}$.
           \item  If both the calls return \yes, then return \yes. Else continue.
       \end{enumerate}
       \item Return \no\
   \end{enumerate}
 
We now prove a claim that will help both for proving correctness and for analyzing the running time of our algorithm.
\begin{claim}\label{claim:sep-instance-size-bound} 
For each $(I',\cP')\in \cF$, it holds that $n_{I'_{L(\cP')}},n_{I'_{R(\cP')}}\leq \lceil 3n/4 \rceil$
\end{claim}
\begin{claimproof}
By Lemma~\ref{lem:family-fully-seperated-instances}, for each $(I',\cP')\in \cF$, the arcs $L(\cP')$ and $R(\cP')$ contain at least $\lfloor n/2 \rfloor$ endpoints of chords in $H(I)$. The total number of endpoints of chords in $H(I)$ is $2n$. As the arcs $L(\cP')$ and $R(\cP')$ are disjoint, they contain at most $2n-\lfloor n/2 \rfloor\leq \lceil 3n/2 \rceil$ endpoints. 
Since we assumed that the chords in $H(I)$ have distinct endpoints, the number of $L(\cP')$-$L(\cP')$ chords in $H(I)$ and the number of $R(\cP')$-$R(\cP')$ chords in $H(I)$ is bounded by $\lceil 3n/4 \rceil$. 

Next recall that $I'$ is a subinstance of $I$ and that $I'_{L(\cP')}$ and $I'_{R(\cP')}$ are the subinstances of $I'$ obtained by deleting all but $L(\cP')$-$L(\cP')$ chords in $H(I')$ and $R(\cP')$-$R(\cP')$ chords in $H(I')$ respectively. This shows that that $n_{I'_{L(\cP')}}$, $n_{I'_{R(\cP')}}\leq \lceil 3n/4 \rceil$.
\end{claimproof}

\paragraph{Correctness.} We prove the correctness of the algorithm by induction on $n$, the number of vertices in the graph in the input instance. If $n\leq 8$, then the algorithm returns the correct answer since it simply brute forces to check whether $I$ is an \yes\ instance or not, and returns \yes\ or \no\ accordingly.
We now prove the correctness of the algorithm for arbitrary $n>8$ assuming, by induction, the correctness of the algorithm on all instances with a graph on $n'$ vertices, where $n'<n$.

By Lemma~\ref{lem:family-fully-seperated-instances}, $I$ is a \yes\ instance if and only if there exists a fully-separated instance $(I',\cP')$ in $\cF$ such that $I'$ is a \yes\ instance.
By Observation~\ref{obs:fully-sep-split}, for any fully-separated instance $(I',\cP')$, $I'$ is a \yes\ instance if and only if both $I'_{L(\cP')}$ and $I'_{R(\cP')}$ are \yes\ instances. Thus we infer that $I$ is a \yes\ instance if and only if there exists a fully-separated instance $(I',\cP')\in \cF$ such that $I'_{L(\cP')}$ and $I'_{R(\cP')}$ are both \yes\ instances. Further by Claim~\ref{claim:sep-instance-size-bound}, for each $(I',\cP')\in \cF$, both instances $I_{L(\cP)}$ and $I_{R(\cP)}$ have at most $\lceil 3n/4 \rceil$ vertices in their graph. Thus, by our inductive assumption, the recursive calls of the algorithm on these instances return the correct answer. So the algorithm in step 3 returns \yes\ if and only if $I$ is an \yes\ instance. If not, in step 4, it correctly returns \no.

\paragraph{Running Time.} Let $T(n)$ denote the running time of the algorithm. Step~1 clearly takes $n^{O(1)}$ time since it is brute force when $n<8$. Step~2, using Lemma~\ref{lem:family-fully-seperated-instances} takes $n^{O(1)}$ time. Further the size of the family $\cF$ returned in step~2 is also $n^{O(1)}$. Step~4 runs in constant time. Using Claim~\ref{lem:family-fully-seperated-instances}, each recursive call in step~3 takes time at most $T(\lceil 3n/4 \rceil)$. 
Combining these, we get the following recurrence:
$$T(n)\leq n^{O(1)} \cdot T(\lceil 3n/4 \rceil) + n^{O(1)}$$
This solves to our required quasi-polynomial running time of $n^{O(\log n)}$
\end{proof}

\subsection{Finding a Family of Fully-Separated Instances (Proof of Lemma~\ref{lem:family-fully-seperated-instances})}
\label{sub:lemmaproof}

Throughout the algorithm, we will make extensive use of the following basic reduction rule that when applied to any instance $I=(G,(H,\alpha),S)$, yields an equivalent subinstance. 

\paragraph{Reduction Rule (\RR).}
Let $v$ be a vertex such that $|S(v)|=1$. Then, delete $v$ from the instance and for each vertex $w$ such that $\alpha(v)$ and $\alpha(w)$ intersect, update its list $S(w):=S(w)\setminus S(v)$. When applying this rule, we maintain the same chord diagram (apart from removing the chords of deleted vertices).
Applying \RR\ exhaustively to an instance means applying \RR\ to the instance, then applying it again to the resulting instance, and so on, until the instance contains no vertex $v$ such that $|S(v)|=1$. See Figure~\ref{fig:reduction_rule} for an example application of \RR.

\begin{observation}\label{obs:RR}
 Applying \RR\ exhaustively to an instance $I$ results in a subinstance $I'$ of $I$ such that $I$ is a \yes\ instance if and only if $I'$ is a \yes\ instance. 
\end{observation}

The following is a simple observation that we repeatedly use, regarding two sets of chords that satisfy some requirements.
\begin{observation}\label{obs:intersecting-chords}
Let $I=(G,(H,\alpha),S)$ and let $X$ and $Y$ be two disjoint non-empty subsets of chords in $I$ such that all chords in $X$ intersect all chords in $Y$. If $I$ is a \yes\ instance, then in a valid 3-coloring of $G$, either all vertices in $\alpha^{-1}(X)$ have the same color or all vertices in $\alpha^{-1}(Y)$ have the same color.
\end{observation}
\begin{proof}
Let $I$ be a \yes\ instance and let $\textsf{col}$ be a valid 3-coloring of $G$. Suppose the claim is not true. Then there exist two vertices, say $x_1,x_2$, in $\alpha^{-1}(X)$ having different colors in $\textsf{col}$ and two vertices, say $y_1,y_2$, in $\alpha^{-1}(Y)$ having different colors in $\textsf{col}$. Since there are precisely three possible colors \{\texttt{red}, \texttt{blue}, \texttt{green}\} for each vertex, by the pigeon hole principle, one vertex among $x_1,x_2$ and one vertex among $y_1,y_2$ must have the same color in $\textsf{col}$. But this is a contradiction since they are adjacent. 
\end{proof}

\begin{figure}[htbp]
  \centering

  \begin{subfigure}[b]{0.31\textwidth}
    \includegraphics[width=\textwidth]{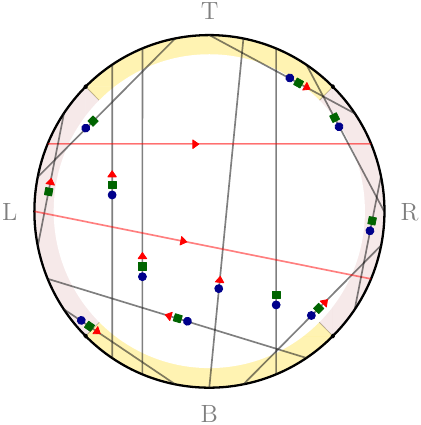}
    \caption{}
    \label{subfig:0a}
  \end{subfigure}
  \hspace{5pt}
  \begin{subfigure}[b]{0.31\textwidth}
    \includegraphics[width=\textwidth]{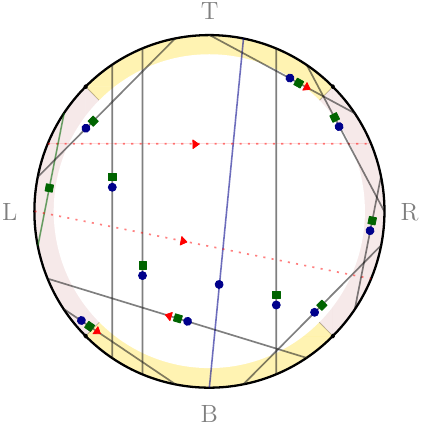}
    \caption{}
    \label{subfig:0b}
  \end{subfigure}
  \hspace{5pt}
  \begin{subfigure}[b]{0.31\textwidth}
    \includegraphics[width=\textwidth]{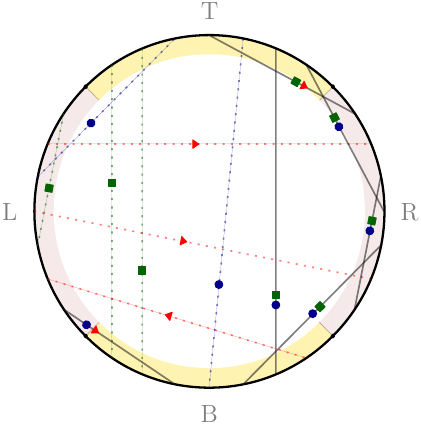}
    \caption{}
    \label{subfig:0c}
  \end{subfigure}
  
  \caption{In Subfigure~\ref{subfig:0a}, as all $L$-$R$ chords intersect all $T$-$B$ chords, Observation~\ref{obs:intersecting-chords} implies that, in any valid 3-coloring, at least one of these two sets contain chords of at most one color. In line with this, all $L$-$R$ chords in Subfigure~\ref{subfig:0a}, are colored \texttt{red}. Subfigure~\ref{subfig:0b} shows the instance after applying \RR\ to all $L$-$R$ chords in the instance, while Subfigure~\ref{subfig:0c} has the instance after \RR\ is applied exhaustively. We use dotted lines to depict deleted chords.}
  \label{fig:reduction_rule}
\end{figure}

We next show how to combine the above observation and \RR\ to obtain a useful algorithm that we use as a helper in our main result. Given an instance $I$ and two disjoint sets of chords $X$, $Y$ in $I$ such that all chords in $X$ intersect all chords in $Y$, the algorithm will help eliminate either all chords in $X$ or all chords in $Y$ from $I$. See Figure~\ref{fig:reduction_rule} for an example of this with $X$ being the set of $L$-$R$ chords and $Y$ the set of $B$-$T$ chords.

\begin{lemma}\label{lem:helper-eliminate-chords}
    There exists an algorithm that takes as input an instance $I=(G,(H,\alpha),S)$ and two disjoint, possibly empty, subsets $X$ and $Y$ of chords in $H(I)$ such that all chords in $X$ intersect all chords in $Y$, runs in time $n^{O(1)}$, and returns a family of at most six subinstances of $I$ such that:
    \begin{itemize}
        \item $I$ is a \yes\ instance if and only if the family contains a subinstance $I'$ of $I$ such that $I'$ is a \yes\ instance.
        \item For each subinstance $I'$ of $I$ in the family, either $H(I')\cap X=\emptyset$ or $H(I')\cap Y=\emptyset$.
    \end{itemize}
\end{lemma}
\begin{proof} Let $I=(G,(H,\alpha),S)$; we first present the algorithm:
\begin{enumerate}
    \item Initialize $\cF:=\emptyset$. 
     \item If there exists two chords in $X$ that intersect, proceed to the next step. Else, for each color $c\in$ \{\texttt{red}, \texttt{blue}, \texttt{green}\} such that for each vertex $v\in \alpha^{-1}(X)$, $c\in S(v)$:
    \begin{enumerate}
        \item Let $I':= I$. Set $S(v)=\{c\}$ in $I'$ for each vertex $v$ in $\alpha^{-1}(X)$. 
        \item Apply \RR\ exhaustively to $I'$ to 
        obtain a subinstance $I''$ of $I'$. 
        \item Add $I''$ to $\cF$.
    \end{enumerate}
    \item If there exists two chords in $Y$ that intersect, proceed to the next step. Else, for each color $c\in$ \{\texttt{red}, \texttt{blue}, \texttt{green}\} such that for each vertex $v\in \alpha^{-1}(Y)$, $c\in S(v)$:
    \begin{enumerate}
        \item Let $I':= I$. Set $S(v)=\{c\}$ in $I'$ for each vertex $v$ in $\alpha^{-1}(Y)$. 
        \item Apply \RR\ exhaustively to $I'$ to 
        obtain a subinstance $I''$ of $I'$. 
        \item Add $I''$ to $\cF$.
    \end{enumerate}
    \item Return $\cF$
\end{enumerate}

It is straightforward to see that the algorithm above runs in $n^{O(1)}$ time and that the size of the family is at most six.
We now show that the family $\cF$ only contains subinstances of $I$.  The $I'$ constructed in any iteration of steps $2(a)$ or $3(a)$ is clearly a subinstance of $I$. This is because we only shorten the list of colors for vertices in $\alpha^{-1}(X)$ or $\alpha^{-1}(Y)$.
Further, we know that exhaustively applying \RR\ to an instance always returns a subinstance of that instance. So any $I''$ obtained from some $I'$ in an iteration of step $2(b)$ or $3(b)$ is a subinstance of $I'$ and so a subinstance of $I$.

\paragraph{Property~2.}Let $I''\in \cF$ be added to $\cF$ in an iteration of step $2(c)$ and obtained by exhaustively applying \RR\ on an instance $I'$ constructed in step $2(a)$ in the same iteration of step 2. Then for each vertex $v$ in $\alpha^{-1}(X)$, the size of the list of colors $|S(v)|$ in $I'$ is one and thus exhaustive application of \RR\ on $I'$ removes it. Consequently $I''$ does not contain any chords from $X$, i.e., $H(I'')\cap X=\emptyset$. By a similar argument, if $I''\in \cF$ was added to $\cF$ in step $3(c)$, $H(I'')\cap Y=\emptyset$. This proves property~2.

\paragraph{Property~1.}
Let $I''$ be a \yes\ instance in $\cF$. We show that $I$ is also a \yes\ instance. Without loss of generality, let $I''$ be obtained in an iteration of step~$2(b)$ by applying \RR\ exhaustively to an $I'$ constructed in step~$2(a)$ of the same iteration of step~$2$. 
%
By Observation~\ref{obs:RR}, we know that if $I''$ is a \yes\ instance, then $I'$ is a \yes\ instance.
Since $I'$ is a \yes\ instance, $I$ is also a \yes\ instance simply by how $I'$ is constructed --- we only shorten the list of colors for vertices in $\alpha^{-1}(X)$. 
We can make a similar argument if $I''$ is obtained in an iteration of step~$3(b)$.

Let $I$ be a \yes\ instance. We now show that there exists $I''\in \cF$ that is also a \yes\ instance. Suppose $X=\emptyset$ or $Y=\emptyset$, then there exist $I''\in \cF$ that was obtained by exhaustively applying \RR\ to $I'=I$ by the working of step 2 and step 3. This $I''$ is guaranteed to be a \yes\ instance if $I$ is a \yes\ instance. 
Therefore, let $X\neq \emptyset$ and $Y\neq\emptyset$, and let \textsf{col} be a valid 3-coloring of $G$. Then by Observation~\ref{obs:intersecting-chords} either all vertices in $\alpha^{-1}(X)$ have the same color or all vertices in $\alpha^{-1}(Y)$ have the same color. Without loss of generality, let all vertices in $\alpha^{-1}(X)$ have the same color $c$. Then the iteration of step~2 with color $c$ constructs a subinstance $I'$, this instance will be a \yes\ instance because the coloring \textsf{col} is a valid 3-coloring for $I'$ as well. Now since $I'$ is a \yes\ instance, the $I''$ constructed by applying \RR\ exhaustively on $I'$ will also be a \yes\ instance. 
\end{proof}

We are now ready to start discussing how to prove Lemma~\ref{lem:family-fully-seperated-instances}, and construct our desired family of fully-separated instances. Recall that a fully-separated instance $(I,\cP)$ has no $L(\cP)$-$R(\cP)$ chords in $H(I)$ and $T(\cP)=B(\cP)=\emptyset$, where $I$ is a problem instance and $\cP$ is a circle partition. 
We start by finding a family of tuples $(I,\cP)$ where each tuple satisfies only the first property among the two properties of fully-separated instances. We call a tuple $(I,\cP)$ a {\em semi-separated} instance if it has no $L(\cP)$-$R(\cP)$ chords in $H(I)$. Observe that a fully-separated instance $(I,\cP)$ is a semi-separated instance with $T(\cP)=B(\cP)=\emptyset$.  We now show how to find a family of semi-separated instances with our desired properties.

The algorithm is fairly simple, we start with a circle partition $\cP$ such that each arc in the partition has nearly equal number of endpoints of chords in $H(I)$ where $I$ is the input instance. We use our helper algorithm from Lemma~\ref{lem:helper-eliminate-chords}, on $I$ along with the set of $L(\cP)$-$R(\cP)$ chords and $T(\cP)$-$B(\cP)$ chords in $H(I)$. This gives us a family $\cF$ of subinstances of $I$, each of which either have no $L(\cP)$-$R(\cP)$ chords or no $T(\cP)$-$B(\cP)$ chords. We now use $\cF$ to construct a family of semi-separated instances. If $I'\in \cF$ does not contain $L(\cP)$-$R(\cP)$ chords then $(I',\cP)$ is semi-separated. But for $I'\in \cF$ that do not contain $T(\cP)$-$B(\cP)$ chords, we rotate the partition $\cP$ counterclockwise so that $T$ becomes $L$ and $B$ becomes $R$, to obtain $\tilde{\cP}$. Then in this case $(I',\tilde{\cP})$ is semi-separated. See Figure~\ref{fig:branching_step_1} below for an example.

\begin{figure}[htbp]
    \centering
    
    \begin{subfigure}[b]{0.31\textwidth}
    \includegraphics[width=\textwidth]{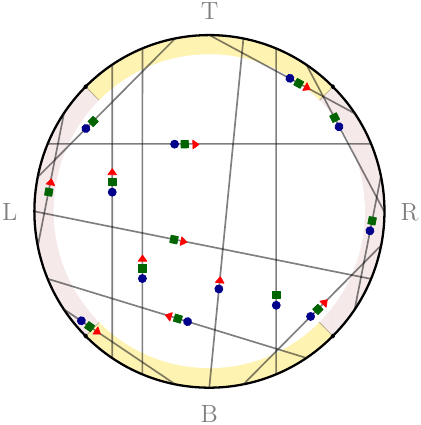}
    \caption{}
    \label{subfig:3a}
    \end{subfigure}
    \hspace{20pt}
    \begin{subfigure}[b]{0.31\textwidth}
    \includegraphics[width=\textwidth]{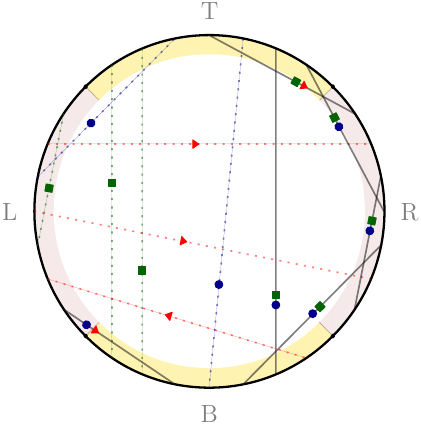}
    \caption{}
    \label{subfig:3c}
    \end{subfigure}
    
    \caption{Subfigure~\ref{subfig:3a} shows a {\sc Circle Graph List 3-Coloring} instance $I$ and a circle partition $\cP=(L,T,R,B)$ with equal number of chord endpoints in each arc. Subfigure~\ref{subfig:3c} shows a semiseparated instance $(I',\cP)$, in the family obtained after applying Lemma~\ref{lem:semi-separated} to $(I,\cP)$. 
    Lemma~\ref{lem:semi-separated} internally calls Lemma~\ref{lem:helper-eliminate-chords} on $I$, $L$-$R$ chords, $T$-$B$ chords. This instance is obtained when Lemma~\ref{lem:helper-eliminate-chords} eliminates $L$-$R$ chords by coloring them \texttt{red}.}
    \label{fig:branching_step_1}
\end{figure}
We now state and prove the lemma which makes this formal. 

\begin{lemma}\label{lem:semi-separated}
    There exists an algorithm that takes as input an instance $I$ of \textsc{Circle Graph List 3-Coloring}, runs in time $n^{O(1)}$, and returns a family containing at most six semi-separated instances such that:
    \begin{enumerate}
        \item $I$ is a \yes\ instance if and only if the family contains a semi-separated instance $(I',P')$ where $I'$ is a \yes\ instance.
         \item For each semi-separated instance $(I',\cP')$ in the family:
        \begin{enumerate}
            \item $I'$ is a subinstance of $I$
            \item The arcs $L(\cP')$ and $R(\cP')$ contain at least $\lfloor n/2 \rfloor$ endpoints of chords in $H(I)$ 
        \end{enumerate}
    \end{enumerate}
\end{lemma}

\begin{proof}
We start with an algorithm that takes as input an instance $I=(G,(H,\alpha),S)$ of \textsc{Circle Graph List 3-Coloring} and returns a family $\mathcal{S}$ of semi-separated instances:
\begin{enumerate}
    \item Construct an arbitrary circle partition $\cP$ whose arcs $L(\cP)$, $T(\cP)$, $R(\cP)$ and $B(\cP)$ contain at least $\lfloor n/2 \rfloor$ endpoints of chords in $I(H)$ and define $\tilde{\cP}$ to be the circle partition with $L(\tilde{\cP}) := T(\cP)$, $T(\tilde{\cP}) := R(\cP)$, $R(\tilde{\cP}) := B(\cP)$ and $B(\tilde{\cP}) := L(\cP)$. 
    \item Let $H_{LR}$ and $H_{TB}$ be the set of $L(\cP)$-$R(\cP)$ chords and $T(\cP)$-$B(\cP)$ chords in $H(I)$ respectively.
    \item Apply Lemma~\ref{lem:helper-eliminate-chords} on $I,H_{LR},H_{TB}$ to obtain a family $\cF$ of at most six subinstances of $I$.
    \item Initialize $\mathcal{S}:=\emptyset$. 
    \item For each $I'\in \cF$:
    \begin{enumerate}
        \item If $H(I')\cap H_{LR}=\emptyset$ then add $(I',\cP)$ to $\mathcal{S}$. If not proceed to (b).
        \item If $H(I')\cap H_{TB}=\emptyset$ then add $(I',\tilde{\cP})$ to $\mathcal{S}$
    \end{enumerate}
    \item Return $\mathcal{S}$
\end{enumerate}

We first show that each tuple in $\mathcal{S}$ is a semi-separated instance. By Lemma~\ref{lem:helper-eliminate-chords}, each $I'\in \cF$ is a subinstance of $I$. 
When $H(I')\cap H_{LR}=\emptyset$, it implies $I'$ has no $L(\cP)$-$R(\cP)$ chord. So when a tuple $(I',\cP)$, $I'\in \cF$ is added to $\mathcal{S}$ in step~5(a), $(I',\cP)$ is a semi-separated instance. 
Next observe that $T(\cP)$-$B(\cP)$ chords in $H(I)$ are the same as $L(\tilde{\cP})$-$R(\tilde{\cP})$ chords in $H(I)$. So when $H(I')\cap H_{TB}=\emptyset$, it implies $I'$ has no $L(\tilde{\cP})$-$R(\tilde{\cP})$ chord.
Thus when a tuple $(I',\tilde{\cP})$, $I'\in \cF$ is added to $\mathcal{S}$ in step~5(b), $(I',\tilde{\cP})$ is a semi-separated instance.
Also, $\mathcal{S}$ has size at most six because $\mathcal{F}$ has size at most six by Lemma~\ref{lem:helper-eliminate-chords}. Finally since $|\cF|\leq 6$ and the algorithm from Lemma~\ref{lem:helper-eliminate-chords} runs in $n^{O(1)}$ time, it is easy to see that the entire algorithm runs in $n^{O(1)}$ time. 

We are left to show that the returned family $\mathcal{S}$ has the desired properties. To this end, note that property~2(b) in the lemma just follows by how $\cP$ and $\tilde{\cP}$ are defined. This combined with the fact that all instances in $\cF$ are subinstances of $I$, implies property~2. To see property~1, observe that by Lemma~\ref{lem:helper-eliminate-chords}, we have that $I$ is a \yes\ instance if and only if there exists $I'\in \cF$ such that $I'$ is a \yes\ instance. Further for each $I'\in \cF$, $I'$ is a subinstance of $I$ and either $H(I')\cap H_{LR}=\emptyset$ or $H(I')\cap H_{TB}=\emptyset$. Consequently, for each $I'\in \cF$ either $(I',\cP)\in \mathcal{S}$ or $(I',\tilde{\cP})\in \mathcal{S}$. These all together imply $I$ is a \yes\ instance if and only if there exists $(I',\cP')\in \mathcal{S}$ such that $I'$ is a \yes\ instance and therefore property~1 holds, which concludes the proof of the lemma.
\end{proof}

To obtain our desired family of fully-separated instances, we first apply Lemma~\ref{lem:semi-separated} on $I$ to get a family $\cS$ of semi-separated instances. Next we transform each semi-separated instance in $\cS$ to a ``nice" family of fully-separated instances and take their union to obtain our desired family. Recall that a semi-separated instance $(I,\cP)$ has no $L(\cP)$-$R(\cP)$ chords but $T(\cP)$ and $B(\cP)$ are not necessarily empty. We now show how to convert a semi-separated instance into a ``nice" family of fully-separated instances by growing the arcs $L(\cP)$ and $R(\cP)$ while shrinking $T(\cP)$ and $B(\cP)$, till the latter arcs become empty. 

\begin{figure}[htbp]
  \centering

  \begin{subfigure}[b]{0.31\textwidth}
    \includegraphics[width=\textwidth]{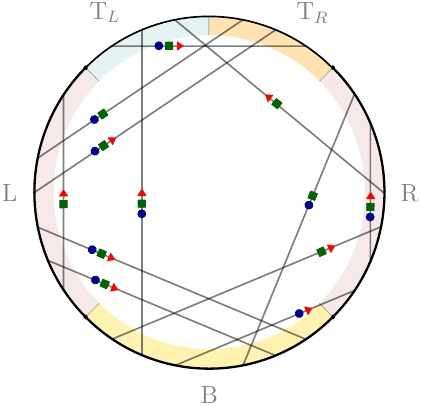}
    \caption{}
    \label{subfig:4a}
  \end{subfigure}
  \hspace{20pt}
  \begin{subfigure}[b]{0.31\textwidth}
    \includegraphics[width=\textwidth]{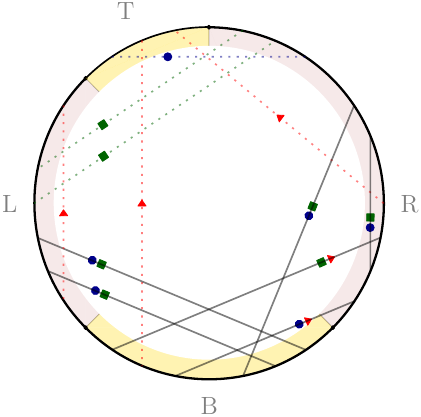}
    \caption{}
    \label{subfig:4c}
  \end{subfigure}
  
  \caption{Subfigure~\ref{subfig:4a} shows a semi-separated instance $(I,\cP)$. Subfigure~\ref{subfig:3c} shows the semiseparated instance $(I',\cP')$, obtained by eliminating $L$-$T_R$ chords, growing the arc $R$ to include $T_R$ and shrinking $T$. Lemma~\ref{lem:semi-fully-sep} internally calls Lemma~\ref{lem:helper-eliminate-chords} on $I$, $L$-$T_R$ chords, $T_L$-$R$ chords. This instance is obtained when Lemma~\ref{lem:helper-eliminate-chords} eliminates $L$-$T_R$ chords by coloring them \texttt{green}.}
  \label{fig:branching_step_2}
\end{figure}

  

The basic idea is to first break one of the arcs $T(\cP)$ or $B(\cP)$; say we break $T(\cP)$. We break it into two parts $T_L$ and $T_R$ such that the number of endpoints of chords in $H(I)$ are distributed equally. Then the sets $H_L$ of $L(\cP)$–$T_R$ chords and $H_R$ of $T_L$–$R(\cP)$ chords are in the desired form: they are disjoint, and every chord in $H_L$ intersects every chord in $H_R$. We now invoke our helper algorithm from Lemma~\ref{lem:helper-eliminate-chords}, on $I$ along with $H_L$ and $H_R$ to obtain a family of subinstances $\cF$ of $I$, such that for each subinstance $I' \in \cF$, the set of chords $H(I')$ contains either no chords from $H_L$ or no chords from $H_R$.

Now for some $I'\in \cF$, say it is the former case of no chords from $H_L$, we can grow the arc $R(\cP)$ in the partition $\cP$ to include $T_R$ and shrink $T(\cP)$ by removing $T_R$. Refer Figure~\ref{fig:branching_step_2} for an example. Using the number of chord endpoints in $T(\cP)$ and $B(\cP)$ as a measure of progress, we recursively invoke the algorithm on the resulting instances till $T(\cP)$ and $B(\cP)$ become empty. Since the measure reduces by a constant factor in each step, the recursion depth is $O(\log n)$, with at most $O(1)$ recursive calls per level. This leads to a total running time of $c^{O(\log n)} = n^{O(1)}$ for some constant $c$. 
We make this formal in the following lemma.
%
\begin{lemma}\label{lem:semi-fully-sep}
    There exists an algorithm that takes as input a semi-separated instance $(I,\cP)$, runs in time $n^{O(1)}$, and returns a family containing at most $n^{O(1)}$ fully-separated instances such that:
    \begin{enumerate}
        \item $I$ is a \yes\ instance if and only if the family contains a fully-separated instance ($I'$,$\cP'$) where $I'$ is a \yes\ instance. 
        \item For each fully-separated instance ($I'$,$\cP'$) in the family:
        \begin{enumerate}
            \item $I'$ is a subinstance of $I$
            \item $L(\cP)\subseteq L(\cP')$ and $ R(\cP)\subseteq R(\cP')$
        \end{enumerate}
    \end{enumerate}
\end{lemma}
\begin{proof}
We first present an algorithm that takes as input a semi-separated instance $(I,\cP)$ and returns a family $\cF$ of fully-separated instances:\\ \\
\textsf{ALG-Semi-Full ($(I,\cP)$)}
\begin{enumerate}
    \item Initialize $\cF:=\emptyset$
    \item Let $X$ be the arc among $T(\cP)$ and $B(\cP)$ having maximum number of endpoints of chords in $H(I)$ and let $x$ be the number of endpoints of chords in $H(I)$ in $X$.
    \item If $x=0$. Define circle partition $\tilde{\cP}$ as $L(\tilde{\cP})=L(\cP)\cup T(\cP)$, $T(\tilde{\cP})=\emptyset$, $R(\tilde{\cP})=R(\cP)\cup B(\cP)$, $B(\tilde{\cP})=\emptyset$. Return $\cF=\{(I,\tilde{\cP})\}$.
    \item If $x=1$, let $\cP':=\cP$:
    \begin{enumerate}
    \item If there are no $L$-$X$ chords in $H(I)$:
    \begin{itemize}
             \item Set $T(\cP')=\emptyset$ and $R(\cP')=R(\cP)\cup T(\cP)$ if $X=T(\cP)$
             \item Set $B(\cP')=\emptyset$ and $R(\cP')=R(\cP)\cup B(\cP)$ if $X=B(\cP)$ 
    \end{itemize}
    Else if there are no $X$-$R$ chords in $H(I)$:
    \begin{itemize}
                \item  Set $T(\cP')=\emptyset$ and $L(\cP')=L(\cP)\cup T(\cP)$ if $X=T(\cP)$ 
                \item Set $B(\cP')=\emptyset$ and $L(\cP')=L(\cP)\cup B(\cP)$ if $X=B(\cP)$ 
    \end{itemize}
    \item Recursively call the algorithm, \textsf{ALG-Semi-Full}, on $(I,\cP')$ and let $\cF'$ be the family of fully-separated instances returned. Set $\cF = \cF'$ and return it.
    \end{enumerate}
    \item If $x>1$:
    \begin{enumerate}
      \item Partition $X$ into arcs $X_L$ and $X_R$, each containing at least $\lfloor x/2 \rfloor$ endpoints of chords in $H(I)$. Let $H_L$ be the set of $L$-$X_R$ chords in $H(I)$ and $H_R$ the set of $X_L$-$R$ chords.
        \item Compute a family $\mathcal{S}$ of at most six subinstances of $I$ using Lemma~\ref{lem:helper-eliminate-chords} on $I,H_L,H_R$
        \item For each $I'\in \mathcal{S}$, let $\cP_{I'}:=\cP$ do:
        \begin{enumerate}
             \item If $H(I')\cap H_L=\emptyset$, i.e. there are no $L$-$X_R$ chords in $H(I')$
             \begin{itemize}
                 \item Set $T(\cP_{I'})=X_L$ and $R(\cP_{I'})=R(\cP)\cup X_R$ if $X=T(\cP)$ 
                 \item Set $B(\cP_{I'})=X_L$ and $R(\cP_{I'})=R(\cP)\cup X_R$ if $X=B(\cP)$
             \end{itemize}
            Else if $H(I')\cap H_R=\emptyset$, i.e. there are no $X_L$-$R$ chords in $H(I')$
            \begin{itemize}
                \item  Set $T(\cP_{I'})=X_R$ and $L(\cP_{I'})=L(\cP)\cup X_L$ if $X=T(\cP)$
                \item Set $B(\cP_{I'})=X_R$ and $L(\cP_{I'})=L(\cP)\cup X_L$ if $X=B(\cP)$
            \end{itemize}
            \item Recursively call the algorithm, \textsf{ALG-Semi-Full}, on $(I',\cP_{I'})$ and let $\cF'$ be the family of fully-separated instances returned. We set $\cF=\cF\cup \cF'$.
        \end{enumerate}
        \item Return $\cF$
    \end{enumerate}
\end{enumerate}

\paragraph{Correctness.} Let $c_{I,\cP}$ be the number of endpoints of chords of $H(I)$ that have an endpoint in $T(\cP)\cup B(\cP)$. We prove the lemma by induction on $c_{I,\cP}$. 
\\ 

\textit{Base case ($c_{I,\cP}=0$).} In this case no chords in $H(I)$ have any endpoint in $T(\cP)\cup B(\cP)$. So we can set $T$ and $B$ to $\emptyset$ and grow $L$ and $R$, this is what we do in step~3. 
Recall $x$ defined in step~2 of the algorithm. Since $c_{I,\cP}=0$, by definition of $x$, we have $x=0$. Clearly the family $\mathcal{F}=\{(I,\tilde{\cP})\}$ returned in step~3 satisfies all the properties of the lemma.
\\

We now prove the correctness of the algorithm for arbitrary $c_{I,\cP}>0$ assuming, by induction, the correctness of the algorithm on all inputs $(\hat{I},\hat{\cP})$ with $c_{\hat{I},\hat{\cP}}<c_{I,\cP}$. For ease, let $c$ denote $c_{I,\cP}$. 
We first show a bound on $c_{I,\cP'}$ and 
$c_{I',\cP_{I'}}$ for each $I'\in \mathcal{S}$. Recall that $\cP'$ is the circle partition constructed from $\cP$ in step~4 for the case when $x=1$.
Also recall that $\cP_{I'}$ is the circle partition constructed from $\cP$ in the iteration of step~5(c) that processes $I'$. 
\begin{observation}
$c_{I,\cP'}\leq c-1$ and for each $I'\in \mathcal{S}$, $c_{I',\cP_{I'}}\leq \lceil 3c/4 \rceil$ and $c_{I',\cP_{I'}}\leq c-1$.
\end{observation}
\begin{proof}
    $\cP'$ is constructed in step~4 when $x=1$. Here the arc $X$ has exactly one endpoint of chords in $H(I)$. Recall that $X$ is either $T(\cP)$ or $B(\cP)$. When we construct $\cP'$ from $\cP$, if $X=T(\cP)$ we set $T(\cP')=\emptyset$ and $B(\cP')=B(\cP)$ and if $X=B(\cP)$ we set $B(\cP')=\emptyset$ and $T(\cP')=T(\cP)$. Thus by definition of $c_{I,\cP'}$ and $c$, it follows that $c_{I,\cP'}\leq c-1$.

    Next we observe that each $I'\in \cS$ is a subinstance of $I$ since $\cS$ is a family constructed using Lemma~\ref{lem:helper-eliminate-chords}. 
    Then recall that $x$ is the number of endpoints of chords in $H(I)$ in $X$ and $X$ is the arc among $T(\cP)$ and $B(\cP)$ having maximum number of endpoints of chords in $H(I)$. So we have $c\leq 2x$ by definition of $c$ and $x$. Let $I'\in \cS$, when construct $\cP_{I'}$, we ensure that either arc $X_L$ or arc $X_R$, which each have at least $\lfloor x/2 \rfloor \geq 1$ endpoints of chords in $H(I)$, is removed from $X$ in $\cP$. Further if $X$ is $T(\cP)$, only $T$ is modified in $\cP'$ and $B$ remains the same. Similarly if $X$ is $B(\cP)$, only $B$ is modified in $\cP'$ and $T$ remains the same.
    This gives us $c_{I',\cP_{I'}}\leq c-1$ and $$c_{I',\cP_{I'}}\leq c - \lfloor x/2 \rfloor \leq c - \lfloor c/4 \rfloor \leq \lceil 3c/4 \rceil$$
\end{proof}

\textit{Case $x=1$.} We now prove correctness for the case when $x=1$. Recall our basic assumption that no two chords share two endpoints. When $x=1$, there either exist an $L$-$X$ chords in $H(I)$ or a $X$-$R$ chord in $H(I)$ but not both. In the former case, we can grow the arc $R(\cP)$ to include $X$ to obtain $\cP'$. In the latter case, we can grow the arc $L(\cP)$ to include $X$ to obtain $\cP'$. In the algorithm each case of $X=T(\cP)$ and $X=B(\cP)$ is explicitly written in step~4(a).
Suppose now that $X=T(\cP)$ and there are no $L$-$X$ chords, we show that $(I,\cP')$ is a semi-separated instance. Here there are no $L(\cP')$-$R(\cP')$ chords because there are no $L(\cP)$-$R(\cP)$ chords and no $L(\cP)$-$T(\cP)$ chords. This is because in this case we set $R(\cP')=R(\cP)\cup T(\cP)$ and $T(\cP')=\emptyset$. A similar argument can be made for the remaining cases.
Combining this with the fact that $c_{I,\cP'}<c-1$, by induction, the family returned by calling the algorithm recursively on $(I,\cP')$ has all the desired properties with respect to $(I,\cP')$. Using this along with how $\cP'$ is defined, it is easy to observe that the family also satisfies all the properties with respect to $(I,\cP)$ that we want for correctness.
\\ 

\textit{Case $x>1$.} We now finally prove that the algorithm is correct even in the case $x>1$.
We first show that each tuple $(I',\cP_{I'})$ is a semi-separated instance. For this, first observe that each $I'\in \cS$ is a subinstance of $I$ by Lemma~\ref{lem:helper-eliminate-chords}. The lemma also guarantees that for each $I'\in \cS$, either $H(I')\cap H_L=\emptyset$ or $H(I')\cap H_R=\emptyset$, i.e., there are no $L$-$X_{R}$ chords or there are no $X_L$-$R$ chords in $H(I')$. So the algorithm must execute step~5(c) and construct a circle partition $\cP_{I'}$. Furthermore, in either case, our construction of $\cP_{I'}$ ensures that there are no $L(\cP_{I'})$-$R(\cP_{I'})$ chords in $H(I')$. By arguments similar to that discussed in case $x=1$, we can conclude that $(I',\cP_{I'})$ is a semi-separated instance.
%

We are left to show that the family $\cF$ has the desired properties. To this end, observe that since $c_{I',\cP_{I'}}<c$, by our inductive assumption, the family $\cF_{I'}$ obtained by applying the algorithm recursively on $(I',\cP_{I'})$ is correct and has all the desired properties with respect to $(I',\cP_{I'})$. Since $\cF=\cup_{I'\in \cS} \cF_{I'}$, each instance in $\cF$ is a fully separated instance. \\[5pt]
\noindent For property~1, recall that Lemma~\ref{lem:helper-eliminate-chords} ensures $I$ is a \yes\ instance if and only if there exists a subinstance $I'\in \cS$ that is a \yes\ instance. By induction, we are guaranteed that $I'$ is a \yes\ instance if and only if there exists a fully-separated instance $(I'',\cP'')\in \cF_{I'}$ such that $I''$ is a \yes\ instance. This combined with how $\cF$ is constructed, implies property~1 --- that $I$ is a \yes\ instance if and only if there exists a fully separated instance $(I'',\cP'')\in \cF$ that is a \yes\ instance.\\[5pt]
\noindent For property 2, note that the family $\cS$ is constructed using  Lemma~\ref{lem:helper-eliminate-chords} in step~5(b) and the lemma guarantees that each $I'\in \cS$ is a subinstance of $I$. Next recall that the family $\cF_{I'}$ is constructed by a recursive call on $(I',P_{I'})$. By construction of $\cP'$ in step~5(c), $L(\cP)\subseteq L(\cP_{I'})$ and $R(\cP)\subseteq R(\cP_{I'})$.
Furthermore, for each $(I'',\cP'')\in \cF_{I'}$, $I'\in \cS$, we are guaranteed by induction that $I''$ is a subinstance of $I'$ and $L(\cP_{I'})\subseteq L(\cP'')$ and $R(\cP_{I'})\subseteq R(\cP'')$. 
 Together, these imply property 2 --- that, $I''$ is a subinstance of $I$ and that $L(\cP)\subseteq L(\cP'')$ and $R(\cP)\subseteq R(\cP'')$. This concludes the proof of correctness of the algorithm.

\paragraph{Running time \& size.} We first observe that all steps other than the recursive calls take $n^{O(1)}$ time. Further $|\cS|\leq 6$ by Lemma~\ref{lem:helper-eliminate-chords}.
Let $T(n,c)$ denote the running time of the algorithm.
When $c<4$, the algorithm runs in $n^{O(1)}$ time. This is because for $c<4$, we have $T(n,c)\leq 6\cdot T(n,c-1)+n^{O(1)}$. 
When $c\geq 4$, we then have $T(n,c)\leq 6\cdot T(n,\lceil 3c/4 \rceil) + n^{O(1)}$.
This evaluates to $T(n,c)\leq n^{O(1)}\cdot 6^{O(\log c)}$.
We know that $c\leq 2n$ since $c$ is upper bounded by the total number of endpoints of chords in $H(I)$ which is at most $2n$. This yields $T(n,c)\leq n^{O(1)}$. A very similar argument can also be made to bound the size of $\cF$ by $n^{O(1)}$.

\end{proof}

With Lemma~\ref{lem:semi-separated} and Lemma~\ref{lem:semi-fully-sep} in hand, we are now ready to formally prove Lemma~\ref{lem:family-fully-seperated-instances}, and thus construct our desired family of fully-separated instances.
\begin{proof}[Proof of Lemma~\ref{lem:family-fully-seperated-instances}] 
We first present our algorithm that takes as input an instance $I=(G,(H,\alpha),S)$ of \textsc{Circle Graph List 3-Coloring} and returns a family $\cF$ of fully-separated instances below:
\begin{enumerate}
    \item Compute a family $\mathcal{S}$ of at most six semi-separated instances using Lemma~\ref{lem:semi-separated} on $I$.
    \item Initialize $\cF:=\emptyset$.
    \item For each semi-separated instance $(I',\cP')\in \mathcal{S}$:
    \begin{enumerate}
        \item Compute a family $\cF_{(I',\cP')}$ of fully-separated instances using Lemma~\ref{lem:semi-fully-sep} on $(I',\cP')$.
        \item Set $\cF=\cF\cup \cF_{(I',\cP')}$
    \end{enumerate}
    \item Return $\cF$
\end{enumerate}

\textit{Fully-separated, size, runtime.} Firstly it is clear that $\cF$ is a family of fully-separated instances since it is a union of families of fully-separated instances obtained by using Lemma~\ref{lem:semi-fully-sep} on various semi-separated instances. Further $|\cS|\leq 6$ and size of each $\cF_{(I',\cP')}$, $(I',\cP')\in \cS$ is $n^{O(1)}$. Thus $|\cF|$ is at most $n^{O(1)}$. Since the algorithms of Lemma~\ref{lem:semi-separated} and Lemma~\ref{lem:semi-fully-sep} run in $n^{O(1)}$ time and $|\cS|\leq 6$, the algorithm clearly runs in time $n^{O(1)}$.

\textit{Property 1.} By Lemma~\ref{lem:semi-separated}, $I$ is a \yes\ instance if and only if there exists a semi-separated instance $(I',\cP')\in \mathcal{S}$ such that $I'$ is a \yes\ instance. Further by Lemma~\ref{lem:semi-fully-sep}, 
for each semi-separated instance $(I',\cP')\in \cS$, $I'$ is a \yes\ instance if and only if there exists a fully-separated instance $(I'',\cP'')\in \mathcal{F}_{(I',\cP')}$ such that $I''$ is a \yes\ instance. Combining this with how $\mathcal{F}$ is constructed we get that $I$ is a \yes\ instance if and only if there exists a fully-separated instance $(I'',\cP'')\in \mathcal{F}$ such that $I''$ is a \yes\ instance. 

\textit{Property 2.} We know from Lemma~\ref{lem:semi-separated} that for any $(I',\cP')\in \cS$, the arcs $L(\cP')$ and $R(\cP')$ contain at least $\lfloor n/2 \rfloor$ endpoints of chords in $H(I)$. Now by Lemma~\ref{lem:semi-fully-sep}, for any $(I'',\cP'')\in \cF_{(I',\cP')}$, $(I',\cP')\in \cS$, we have $L(\cP')\subseteq L(\cP'')$ and $R(\cP')\subseteq R(\cP'')$. Thus for any $(I'',\cP'')\in \cF$, the arcs $L(\cP'')$ and $R(\cP'')$ contain at least $\lfloor n/2 \rfloor$ endpoints of chords in $H(I)$. Finally it also follows from Lemma~\ref{lem:semi-separated} and Lemma~\ref{lem:semi-fully-sep} that for each $(I'',\cP'')\in \cF$, $I''$ is a subinstance of $I$.
\end{proof}

\section{Concluding Remarks}
We gave an algorithm running in time $n^{O(\log n)}$ for $3$-coloring circle graphs.
The provided algorithm can easily be made constructive, i.e., can either output a $3$-coloring of a given circle graph or correctly identify that no such coloring exists.
Also, for clarity of exposition, we have presented the algorithm in a form that uses quasi-polynomial space. We note that with minor modifications, Lemma~\ref{lem:semi-separated} and Lemma~\ref{lem:semi-fully-sep}, and consequently Lemma~\ref{lem:family-fully-seperated-instances}, can be adapted to return the elements iteratively rather than all at once. Hence, the algorithm can be implemented using only polynomial space.

Apart from the classical applications of graph coloring, the result also allows us to compute $3$-page book embeddings on general graphs equipped with a fixed vertex ordering. 
Our new algorithm is simpler than several recent algorithmic results targeting the computation of book embeddings~\cite{BekosLGGMR20,GanianMOPR24}, but it is non-trivial. Crucially, it is self-contained, reproducible, and provides a strong indication that the problem is not \NP-hard. 

Whether \textsc{Circle Graph (List) $3$-Coloring} admits a polynomial-time algorithm or not remains as an important open question, and we hope that our results can help pave the way towards finally settling it. This could be resolved either via a polynomial algorithm, or by establishing a quasi-polynomial lower bound akin to those that have been developed for some other problems in the literature~\cite{AaronsonIM14,BravermanKW15,BravermanKRW17}.

\subsection*{Acknowledgments.}
Robert Ganian acknowledges support from the Austrian Science Fund (FWF) [Projects 10.55776/Y1329 and 10.55776/COE12] and the WWTF Vienna Science and Technology Fund [Project 10.47379/ICT22029]. 
Vaishali Surianarayanan acknowledges support from the UCSC Chancellor’s Postdoctoral Fellowship. 
Ajaykrishnan E S and Daniel Lokshtanov acknowledge support from NSF Grant~No.~2505099, \emph{Collaborative Research: AF Medium: Structure and Quasi-Polynomial Time Algorithms}. 
The authors gratefully acknowledge that the problem addressed in this paper was first brought to their attention during a talk at the Workshop on Graph Classes, Optimization and Width Parameters (GROW 2024)~\cite{GROW2024}, and that some ideas were developed during the Dagstuhl Seminar New Tools in Parameterized Complexity: Paths, Cuts, and Decomposition~\cite{Dagstuhl24411}.

\bibliographystyle{plainurl}
\bibliography{ref}
\end{document}